\documentclass[aps,prl,reprint,twocolumn]{revtex4-2}
\usepackage[T1]{fontenc}
\usepackage[utf8]{inputenc}
\usepackage{lmodern}
\usepackage{amsmath,amssymb,bm,amsthm,mathtools}
\usepackage{graphicx}
\usepackage{xcolor}
\usepackage[colorlinks=true,linkcolor=blue,citecolor=blue,urlcolor=blue]{hyperref}
\usepackage{csvsimple}
\usepackage{longtable}
\usepackage{booktabs}
\usepackage{pgfplotstable}

\newcommand{\D}{\mathcal{D}}
\newcommand{\per}{\operatorname{per}}
\newcommand{\Tr}{\operatorname{Tr}}
\newcommand{\ket}[1]{|#1\rangle}
\newcommand{\bra}[1]{\langle#1|}
\newtheorem{theorem}{Theorem}
\begin{document}

\title{How additional labels affect multiparticle indistinguishability of identical particles}
\author{Valery Shchesnovich}
\affiliation{Centro de Ci\^encias Naturais e Humanas, Universidade Federal do ABC, Santo Andr\'e, SP, 09210-170, Brazil}
\date{\today}

\begin{abstract}
For two identical particles, adding labels reduces the total overlap of their label states and decreases their exchange indistinguishability as measured by Hong--Ou--Mandel interference. For more than two particles, the opposite effect can occur: introducing additional labels can boost multiparticle indistinguishability. We present an explicit seven-particle example in which additional labels increase multiparticle indistinguishability by $7.38\%$, as quantified by the symmetric projection weight of the seven-particle label state, while decreasing the mean squared pairwise overlap by nearly half. We trace this effect to the cancellation of collective multiparticle phases. Phases associated with one set of label states cancel those associated with the other, thereby turning destructive contributions to the symmetric weight into constructive ones. We also give sufficient conditions that rule out this effect. For any number of identical particles, additional label states independent of the original label states cannot increase indistinguishability when the original symmetric weight is at least $1/2$. Independent additional label states also cannot increase the symmetric weight of any number of identical particles prepared independently in the same mixed label state.
\end{abstract}
\maketitle

\paragraph{Introduction.}
Multiphoton interference is a resource for linear optical quantum computing and boson sampling~\cite{LOC,AA,RevLOC}.
These applications require single photons that are nearly indistinguishable.
The visibility of the Hong--Ou--Mandel (HOM) dip at a balanced beam splitter measures pairwise exchange indistinguishability~\cite{HOM} and is routinely used to characterize photons from various sources~\cite{QDs,QDsHOM,QDsHOM2,R2Ds}.
Differences in arrival time, spectrum, or polarization reduce the HOM visibility~\cite{HOM}, so our intuition suggests that adding another label to distinguish identical particles should reduce their indistinguishability. For more than two identical particles, however, interference also depends on collective multiparticle phases~\cite{3phPhase,nphPhases,4phphase}, which pairwise interference cannot detect.
The collective phases arise from cyclic exchanges of more than two particles, which contribute complex-valued terms to the symmetric projection weight of the combined multiparticle label state. Consequently, additional labels can affect these phases and increase multiparticle indistinguishability while reducing the pairwise exchange indistinguishabilities.

We study an explicit example of this counterintuitive effect for seven identical particles.
Particle labels are quantum states of the internal degrees of freedom that remain invariant under the action of linear interferometers. Our construction exhibits an increase in multiparticle indistinguishability when additional labels are introduced, even though the indistinguishability of every pair either decreases or remains unchanged.
For example, the label states of photons interfering in a polarization-insensitive spatial interferometer describe their spectral shapes and polarizations.
Our specific construction realizes Drury's counterexample to a permanent inequality proposed by Bapat and Sunder~\cite{Drury2016}.
The same matrices underlie anomalous boson bunching~\cite{BCount,UnfBunch2026}.
We study the role of collective-phase cancellation in this indistinguishability enhancement. Adding labels changes the aggregate $\mathbb Z_2$ phases of higher-order cycle-type contributions to the symmetric projection weight of the combined label state, converting negative contributions into positive ones.
We then derive a general sufficient condition that confines the effect to particles with substantial distinguishability.
Independent additional labels cannot increase the indistinguishability of any original label state whose symmetric weight is at least $1/2$.
The effect is also excluded for identical particles prepared independently in the same mixed label state.

\paragraph{Indistinguishability measure and label states.}
Consider $n\ge 2$ identical particles in a linear interferometer, with one particle per occupied input port.
The interferometer acts on the spatial modes and leaves the internal degrees of freedom, such as temporal mode and polarization, unchanged.
Let $\hat{\varrho}$ be the joint label density operator on $\mathcal H_{\mathrm{int}}^{\otimes n}$, with tensor factors indexed by the occupied input ports~\cite{Shch2015}.
The weight of the symmetric projection of the label state is the projective multiparticle indistinguishability measure \cite{Shch2015A}
\begin{equation}
 \D:=\Tr(\hat S\hat{\varrho}),\qquad
 \hat S:=\frac{1}{n!}\sum_{\sigma}\hat P_\sigma .
 \label{eq:D}
\end{equation}
Here $\hat P_\sigma$ is the unitary representation of the symmetric group on the tensor-power Hilbert space $\mathcal H_{\mathrm{int}}^{\otimes n}$, and $\hat S$ projects onto its symmetric subspace. The symmetric weight $\D$ measures how close bosons or fermions are to their respective ideally indistinguishable states in the degrees of freedom on which linear interferometers act. In the ideal case, the state is completely symmetric for bosons and completely antisymmetric for fermions. For example, identical pure label states give $\D=1$, and a product of mutually orthogonal label states gives $\D=1/n!$.
For fixed interferometer input occupations, $1-\D$ also bounds the total variation distance between the observed output distribution and the ideal indistinguishable-particle distribution~\cite{Shch2015A}.

Below we consider particles prepared independently in pure label states, i.e., $\hat{\varrho}=\bigotimes_{j=1}^n\ket{\psi_j}\bra{\psi_j}$. In this case
Eq.~\eqref{eq:D} becomes
\begin{equation}
 \begin{aligned}
 \D(G)&=\frac{\per G}{n!},\qquad G_{ij}=\langle\psi_i|\psi_j\rangle,\\
 \per G&=\sum_{\sigma}\prod_{i=1}^n G_{i,\sigma(i)}.
 \end{aligned}
 \label{eq:permanent}
\end{equation}
Let $\ket{f_j}$ be the original label state of particle $j$ (the $B$-label), and let $\ket{\phi_j}$ be its state in an additional degree of freedom (the $A$-label).
For the combined labels we have
\begin{equation}
 \ket{\psi_j}=\ket{\phi_j}\otimes\ket{f_j},\qquad G=A\circ B ,
 \label{eq:labels}
\end{equation}
where $A_{ij}=\langle\phi_i|\phi_j\rangle$, $B_{ij}=\langle f_i|f_j\rangle$, and $\circ$ denotes the entrywise (Hadamard) product.
Both Gram matrices are positive semidefinite with unit diagonal. It follows that
\begin{equation}
 |G_{ij}|^2=|A_{ij}|^2|B_{ij}|^2\leq |B_{ij}|^2.
 \label{eq:pairs}
\end{equation}
For two particles with pure label states of both types, the HOM visibility is $|G_{12}|^2$ and the symmetric weight is $\D(G)=(1+|G_{12}|^2)/2$.
Both decrease strictly if $|A_{12}|<1$ and $|B_{12}|>0$.

\paragraph{Indistinguishability boost by adding labels.}
Consider, for example, single photons. Let the $B$-labels describe the temporal shapes and the $A$-labels the polarizations. We will use two normalized temporal modes $\ket{h_1}$ and $\ket{h_2}$ with a real overlap $\rho=\langle h_1|h_2\rangle$.
We take the following temporal $B$-label states for the individual photons:
\begin{align}
 \ket{f_1}&=\ket{h_1},\qquad \ket{f_2}=\ket{h_2},\nonumber\\
 \ket{f_k}&=\frac{\ket{h_1}+e^{-i\theta_k}\ket{h_2}}
 {\sqrt{2+2\rho\cos\theta_k}},\quad k=3,\ldots,7,
 \label{eq:temporal}\\
 (\theta_3,\ldots,\theta_7)&=\tfrac{2\pi}{5}(2,1,0,-1,-2).\nonumber
\end{align}
We compare identical polarizations ($|\phi_i\rangle = |\phi\rangle$, i.e., $A_{ij}=1$ for every pair) with the following states in the orthonormal polarization basis $\{\ket{H},\ket{V}\}$:
\begin{align}
 \ket{\phi_1}&=\ket{H},\qquad \ket{\phi_2}=\ket{V},\nonumber\\
 \ket{\phi_k}&=\frac{\ket{H}+e^{i\theta_k}\ket{V}}{\sqrt{2}},
 \quad k=3,\ldots,7.
 \label{eq:polarization}
\end{align}
At $\rho=0$, the two temporal modes are orthogonal, $B(0)=A^T$, and $G_{ij}(0)=|A_{ij}|^2$ for the polarizations in Eq.~\eqref{eq:polarization}. This case realizes
Drury's counterexample to a permanent inequality~\cite{Drury2016}
\begin{equation}
 \per G(0)=\frac{6185}{128}>\per B(0)=45 .
 \label{eq:exactzero}
\end{equation}
The symmetric weight rises from $\D(B(0))=45/7!\approx0.008929$ to $\D(A\circ B(0))=6185/(128\,7!)\approx0.009587$, a relative increase of $7.38\%$.
Across all $21$ pairs, the mean squared pairwise overlap $\mathcal{V}\equiv \frac{1}{21}\sum_{i<j} |\langle\psi_i|\psi_j\rangle|^2 $ falls from $\mathcal{V}(B(0))=5/12\approx0.4167$ to $\mathcal{V}(A\circ B(0))=25/112\approx0.2232$. No squared pairwise overlap increases.

\begin{figure*}[t]
\centering
\includegraphics[width=0.4\textwidth]{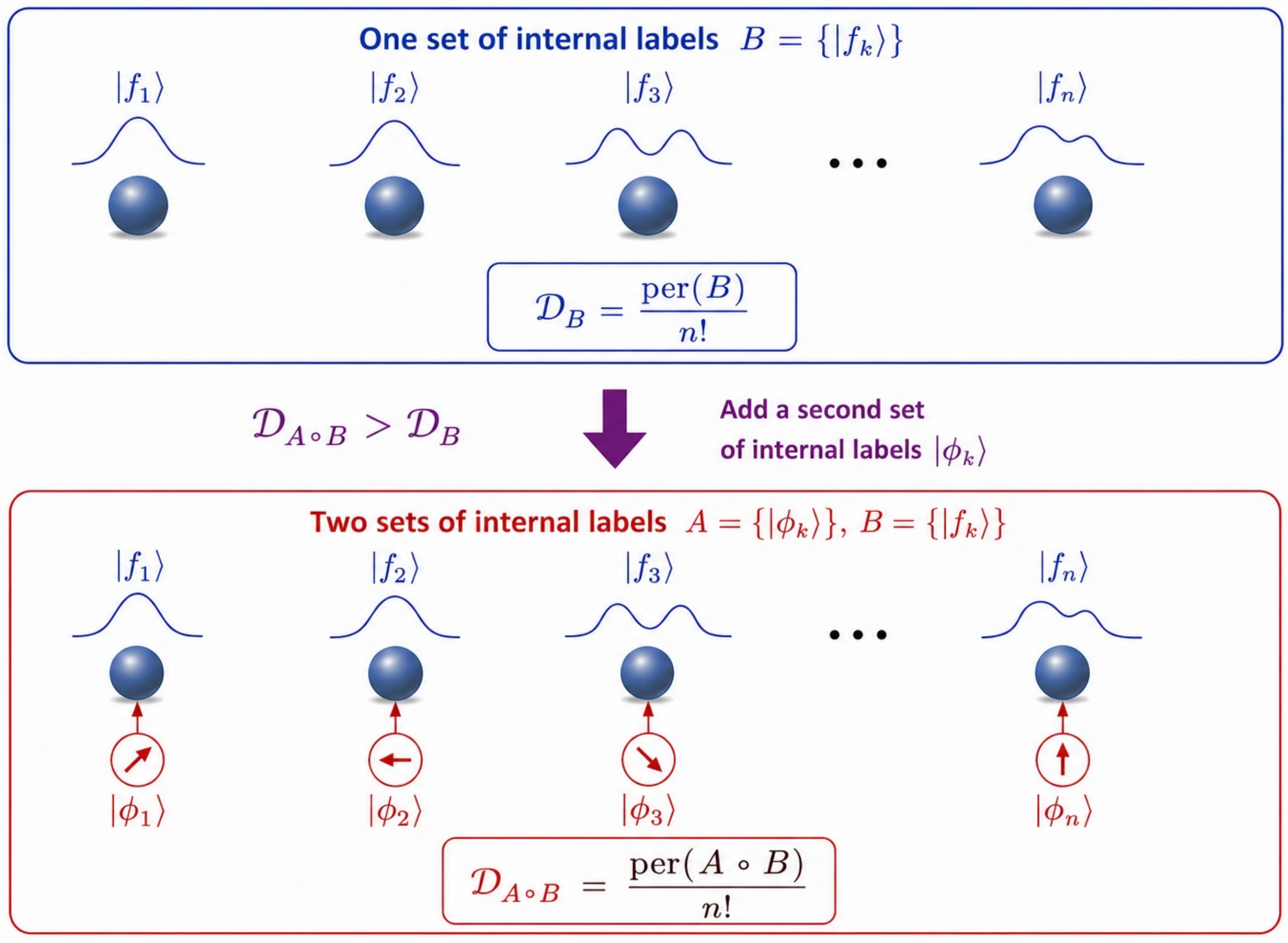}\hspace{0.035\textwidth}
\includegraphics[width=0.4\textwidth]{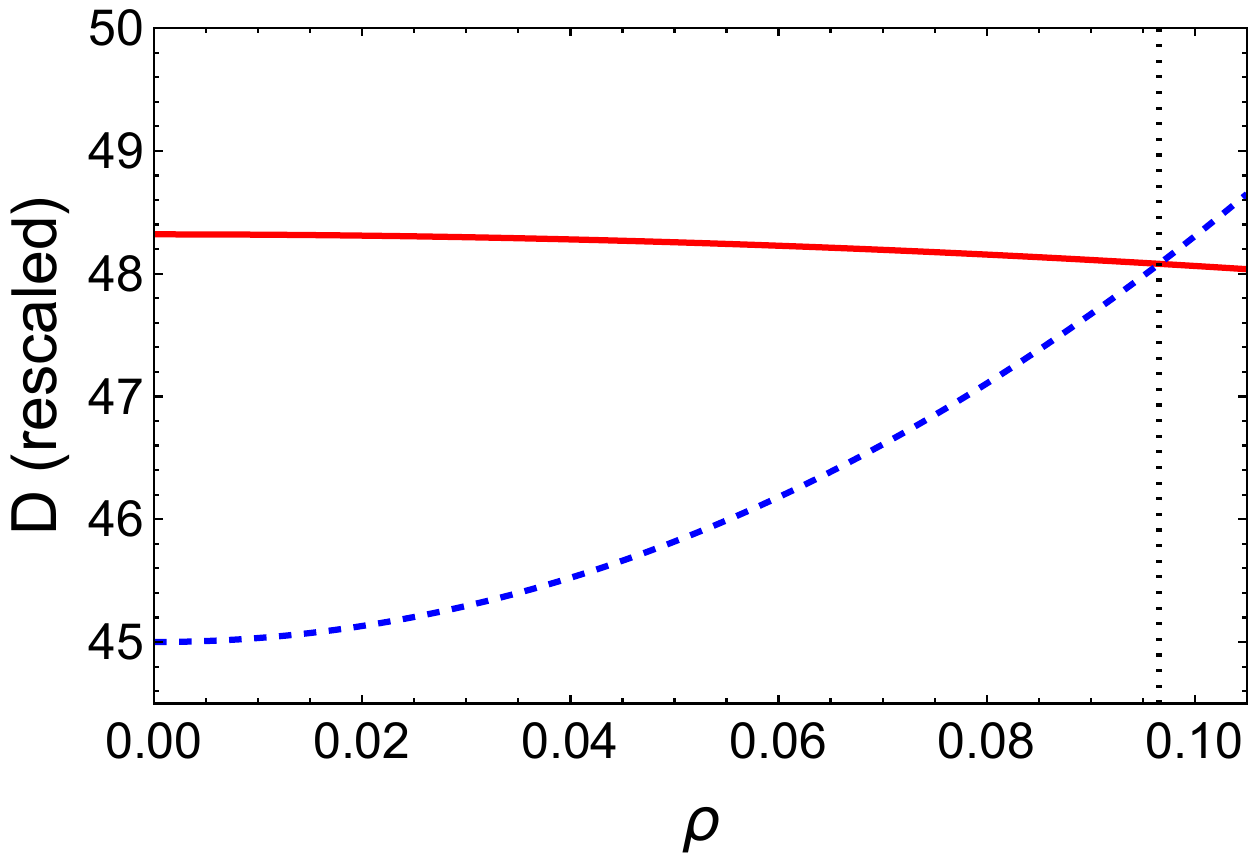}
\caption{The left panel shows the original temporal labels ($B$) and the additional polarization labels ($A$).
The right panel shows $7!\D$ for identical polarizations, $\per B(\rho)$ (blue dashed line), and for the polarizations in Eq.~\eqref{eq:polarization}, $\per(A\circ B(\rho))$ (red solid line).
The vertical dotted line marks the crossing at $\rho_*$.}
\label{fig1}
\end{figure*}

The enhancement persists for nonorthogonal temporal modes. The exact expressions for the permanents derived in Appendix~\ref{appA} give
\begin{equation}
 \begin{gathered}
 \D(A\circ B(\rho))>\D(B(\rho))
 \quad\Longleftrightarrow\quad 0\leq\rho<\rho_*,\\
 \rho_*=0.096583467644\ldots\, .
 \end{gathered}
 \label{eq:threshold}
\end{equation}
Figure~\ref{fig1} shows the behavior of the two symmetric projection weights up to the crossing point.

\paragraph{Collective phases and the indistinguishability boost.}
Recall that every permutation in the symmetric group decomposes into disjoint cycles, each cyclically permuting a subset of the particles.
For a cycle $\nu=(i_1,i_2,\ldots,i_r)$ acting as $i_k\to i_{k+1}$ with $i_{r+1}\equiv i_1$, the cycle weight in the permanent of the Gram matrix $G$ in Eq.~(\ref{eq:permanent}) is $w_G(\nu)=\prod_{a=1}^r G_{i_a,i_{a+1}}$.
If this product is nonzero, its $U(1)$-invariant collective phase reads~\cite{nphPhases}
\begin{equation}
 \Phi_\nu(G)=\arg \left(w_G(\nu)\right)
 =\sum_{a=1}^r\arg G_{i_a,i_{a+1}}\pmod{2\pi}.
 \label{cyclePh}
\end{equation}
A transposition has weight $|G_{ij}|^2$. For $r\ge 3$, a cycle and its inverse have conjugate weights and together contribute $2\operatorname{Re}\left(w_G(\nu)\right)$ when the other particles are fixed. This real contribution can have either sign.

For the label states of Eq.~\eqref{eq:labels}, we also have $w_G(\nu)=w_A(\nu)w_B(\nu)$.
The added labels can reduce the magnitudes and shift the phases.
At $\rho=0$, $w_B(\nu)=w_A(\nu)^*$, so
\begin{equation}
 w_G(\nu)=|w_A(\nu)|^2\geq0 .
 \label{eq:cancel}
\end{equation}
Thus, at $\rho=0$, the polarization phases cancel the temporal phases exactly. Previously destructive cycle contributions therefore become constructive, while the pairwise exchange terms decrease in magnitude.
To determine the net effect of the large number of individual cycle terms, we group them by cycle type, given by the partition $\lambda$ of $n$ formed by their cycle lengths. This gives
\begin{equation}
 \begin{gathered}
 \D(G)=\frac{1}{n!}\sum_{\lambda\vdash n}d^\lambda(G),\\
 d^\lambda(G):=\sum_{\mathrm{type}(\sigma)=\lambda}
                         \prod_{i=1}^n G_{i,\sigma(i)} .
 \end{gathered}
 \label{perG}
\end{equation}
Each cycle type is closed under inversion, so $d^\lambda(G)$ is real. Write $s_G^\lambda(\rho)=\operatorname{sgn}d^\lambda(G(\rho))$ for its sign. For a nonzero sum, this sign is its aggregate $\mathbb Z_2$ phase.
At $\rho=0$, the only negative sums for the $B$-labels alone ($G = B(0)$) are those for $\lambda\in\{(7), (6,1), (5,2), (5,1,1)\}$.
All four change sign when the polarization labels are added, $G = A\circ B(0)$.
The cycle type $\lambda=(4,1,1,1)$ has a positive contribution to the symmetric projection weight that rises from $d^\lambda(B(0))=5/4$ to $d^\lambda(A\circ B(0))=5$.
Among its 105 inverse pairs of four-cycles, 45 change their contributions to the symmetric weight from negative to positive, and the contributions of 10 already positive pairs increase. These gains in the symmetric weight outweigh the decreases due to the other cycle types, including those in the pairwise exchange terms (see Appendix~\ref{appB} for details).

As $\rho$ departs from zero, the overlaps between $A$-label states remain fixed, while those between $B$-label states change in both magnitude and phase.
For instance, for $\rho$ near zero, $d^{(4,1,1,1)}(B(\rho))$ increases while $d^{(4,1,1,1)}(A\circ B(\rho))$ decreases.
At the crossing $\rho =\rho_*$, the four initially negative $B$-label (temporal) cycle types still have negative contributions, and their counterparts with the additional $A$-labels (polarizations) still have positive contributions.
The crossing therefore depends on the magnitudes of the cycle-type contributions as well as their signs.
How to account for this detailed balance for a general collection of label states remains an open problem.

\paragraph{Signatures of the indistinguishability boost.}
The indistinguishability boost shows up in the respective interferometer output probability distributions. For instance, the probability that all seven photons bunch into any one of the seven output ports of an interferometer with transition amplitudes $U_{kj}$ is~\cite{Tichy2015}
\begin{equation}
 P_{b}= 7!\D\sum_{j=1}^7\prod_{k=1}^7 |U_{kj}|^2 .
 \label{eq:fullbunch}
\end{equation}
For a $7\times7$ Fourier interferometer with $U_{kj} = \frac{1}{\sqrt{7}}e^{\frac{2i\pi}{7}kj}$, $0\le k,j\le 6$, this becomes $P_b = \frac{7!}{7^6}\D\approx 0.04284\D$.
At $\rho=0$, the probabilities for the two preparations therefore differ by $P_b(A\circ B(0)) - P_b(B(0)) \approx 2.82222\times10^{-5}$.
\paragraph{Ruling out an indistinguishability boost.}
Let us now explore the limits on the indistinguishability boost obtained by adding labels to particles. Independently emitted photons from a stable source \mbox{\cite{QDs,QDsHOM,QDsHOM2,R2Ds}} can be modeled using the same single-photon label state $\hat{\rho}_1$, which is generally mixed. In this case, the original $n$-particle label state is $\hat{\varrho}_B=\hat{\rho}_1^{\otimes n}$.
For a label system $X$, write $J_X(\sigma):=\Tr(\hat P_\sigma^{(X)}\hat{\varrho}_X)$ for its indistinguishability function \cite{Shch2015} and $\D_X=(1/n!)\sum_{\sigma}J_X(\sigma)$ for its symmetric weight.
The tensor-power label state $\hat{\varrho}_B$ gives~\cite{Shch2015}
\begin{equation}
 J_B(\sigma)=\prod_{\nu\in\operatorname{cycles}(\sigma)}
                   \Tr(\hat{\rho}_1^{|\nu|})\geq0 .
 \label{eq:iid}
\end{equation}
Thus all cycle weights for the original $B$-labels are real and positive.
Suppose the added $n$-particle $A$-labels are independent of the $B$-labels.
We then obtain
\begin{equation}
 J_{AB}(\sigma)=J_A(\sigma)J_B(\sigma).
 \label{modJ}
\end{equation}
For pure product $A$-labels, $J_A(\sigma)=\prod_i A_{i,\sigma^{-1}(i)}$.
More generally, for any $\hat{\varrho}_A$, unitarity gives $|J_A(\sigma)|\le1$.
Since $J_A(\sigma^{-1})=J_A(\sigma)^*$ and $J_B(\sigma^{-1})=J_B(\sigma)$, pairing inverse permutations yields
\begin{equation}
 \begin{aligned}
 \D_{AB}&=\frac{1}{n!}\sum_{\sigma}
                 J_B(\sigma)\operatorname{Re}\left\{J_A(\sigma)\right\}\\
 &\leq\frac{1}{n!}\sum_{\sigma}J_B(\sigma)=\D_B .
 \end{aligned}
 \label{eq:noboost}
\end{equation}
This argument allows arbitrary correlations among the particles' $A$-labels and applies whenever all $J_B(\sigma)$ are real and nonnegative.
By linearity, the bound also holds for classically correlated sources whose total label state is $\hat{\varrho}_{AB}=\sum_s p_s\,\hat{\varrho}_{A|s}\otimes\hat{\rho}_{1|s}^{\otimes n}$, with $\D_B$ evaluated on the $B$-label marginal.

The preceding result covers convex mixtures in which the $B$-labels are independently prepared copies within each mixture component.
We now give a sufficient condition that excludes an indistinguishability boost for an arbitrary $B$-label state, including states with entanglement among the particles.
\begin{theorem}\label{thm:noincrease}
If the $B$-labels and $A$-labels of identical particles are prepared independently and $\D_B\geq 1/2$, then $\D_{AB}\leq\D_B$.
The inequality is strict if $\D_B>1/2$ and $\D_A<1$.
\end{theorem}
\begin{proof}
Regroup the tensor factors of the $n$-particle internal Hilbert space by label type
\begin{equation}
 (\mathcal H_A\otimes\mathcal H_B)^{\otimes n}
 \to \mathcal H_A^{\otimes n}\otimes\mathcal H_B^{\otimes n}.
 \label{regroupH}
\end{equation}
Because the two label systems are independent, the combined label state in this ordering is $\hat{\varrho}_{AB} = \hat{\varrho}_A\otimes \hat{\varrho}_B$.
Let $\hat S_X=(1/n!)\sum_\sigma\hat P_\sigma^{(X)}$ and $\hat Q_X=\hat I_X-\hat S_X$ for $X=A,B$. The symmetric projector for the combined labels becomes
\[
 \hat S_{AB}=\frac{1}{n!}\sum_{\sigma}
                  \hat P_\sigma^{(A)}\otimes\hat P_\sigma^{(B)} .
\]
Since $\hat S_X\hat P_\sigma^{(X)}=\hat S_X$, we get~\cite{IntJQuntInf2015}
\[
(\hat S_A\otimes\hat I_B)\hat S_{AB}=(\hat I_A\otimes\hat S_B)\hat S_{AB}=\hat S_A\otimes\hat S_B.
\]
Thus $\hat S_{AB}$ vanishes on the ranges of
$\hat S_A\otimes\hat Q_B$ and $\hat Q_A\otimes\hat S_B$.
In each of these subspaces, one label type is symmetric and the other
lies in the orthogonal complement of its symmetric subspace.
On the range of $\hat S_A\otimes\hat S_B$, both label types are
symmetric, so $\hat S_{AB}$ acts as the identity.
Since $\hat S_{AB}$ commutes with $\hat Q_A\otimes\hat Q_B$,
its restriction to the range of $\hat Q_A\otimes\hat Q_B$ is
an orthogonal projector. Hence
\[
 \hat S_{AB}\leq \hat S_A\otimes\hat S_B+\hat Q_A\otimes\hat Q_B .
\]
Taking the expectation value in the product state $\hat{\varrho}_A\otimes\hat{\varrho}_B$ gives
\[
 \D_{AB}\leq\D_A\D_B+(1-\D_A)(1-\D_B).
\]
Subtracting $\D_B$ yields
\begin{equation}
 \D_{AB}-\D_B\leq(1-\D_A)(1-2\D_B),
 \label{DAB_DB}
\end{equation}
which proves both assertions.
\end{proof}

Consider again the example of independently prepared single-particle $B$-labels, with $\hat{\varrho}_B=\bigotimes_{k=1}^n\hat{\rho}_k$. Let $\hat{\Pi}=\ket{{\mathrm{ref}}}\bra{{\mathrm{ref}}}$ project onto a normalized reference state.
The squared fidelity with this reference state is $\mathcal F_k:=\Tr(\hat{\Pi}\hat{\rho}_k)$. Since the symmetric projector satisfies $\hat S_B \geq \hat{\Pi}^{\otimes n}\hat{S}_B \hat{\Pi}^{\otimes n}=\hat{\Pi}^{\otimes n}$, we obtain the lower bound
\begin{equation}
 \D_B\geq\prod_{k=1}^n\mathcal F_k.
 \label{eq:fidelity-weight}
\end{equation}
The bound in Eq.~\eqref{eq:fidelity-weight} holds for arbitrary single-particle label states, including those with nonzero coherences between the reference state and its orthogonal complement.
It therefore follows that
\begin{equation}
 \prod_{k=1}^n\mathcal F_k\geq\tfrac12
 \label{DboundF}
\end{equation}
suffices to exclude a boost by Theorem~\ref{thm:noincrease}. This condition requires $\mathcal{F}_k\ge \frac12$ for every $k$ and is satisfied, for example, if $\mathcal F_k\geq 2^{-1/n}$ for every $k$.

For independently prepared particles, we now express the fidelity criterion in terms of the pairwise HOM visibilities~\cite{HOM} $\mathcal V_{kl}:=\Tr(\hat{\rho}_k\hat{\rho}_l)$.

If each single-particle label state is block diagonal with respect to the pure reference state, then
\begin{equation}	
 \hat{\rho}_k=\mathcal F_k \hat{\Pi}+(1-\mathcal F_k)\hat{\rho}_k^\perp,
 \qquad \hat{\Pi}\hat{\rho}_k^\perp=0, 
 \label{1B-labels}
\end{equation}
where $\hat{\rho}_k^\perp$ is a density operator on the orthogonal complement whenever $\mathcal F_k<1$. The residual overlap $\Tr(\hat{\rho}_k^\perp\hat{\rho}_l^\perp)$ lies in $[0,1]$, which gives
\begin{equation}
 \mathcal F_k\mathcal F_l\leq \mathcal V_{kl}
 \leq\mathcal F_k\mathcal F_l+(1-\mathcal F_k)(1-\mathcal F_l).
 \label{Vkl}
\end{equation}
Assuming $\mathcal F_k\geq1/2$ for every $k$, the upper bound in Eq.~\eqref{Vkl} gives the sufficient condition
\begin{equation}
 \min_{k<l}\mathcal V_{kl}\geq 1 - 2f_n(1-f_n),\quad f_n:= 2^{-1/n}
 \label{eq:uniform-visibility}
\end{equation}
(see Appendix~\ref{appC} for details). Within this block-diagonal model, the uniform threshold is sharp as a sufficient condition for Eq.~\eqref{DboundF}. Equality is attained by identical rank-two states $\hat{\rho}_k=\mathcal{F}\hat{\Pi}+(1-\mathcal{F})\ket{g}\bra{g}$, where $\hat{\Pi}\ket{g}=0$ and $1/2\leq \mathcal{F}<1$. Equation~\eqref{eq:uniform-visibility} is a sufficient condition for excluding an indistinguishability boost in terms of the pairwise HOM visibilities.

An alternative sufficient condition based on the pairwise visibilities $\mathcal{V}_{kl}$ applies to arbitrary independently prepared mixed single-particle label states. Take one particle, say particle $1$, and write its $B$-label state as $\hat{\rho}_1=\sum_s q_s\ket{e_s}\bra{e_s}$. We now apply Eq.~\eqref{eq:fidelity-weight} separately to each component, using $\ket{e_s}$ as the reference state $|\mathrm{ref}\rangle$. Linearity of $\D_B$ under this convex decomposition and the inequality $\prod_j u_j\geq1-\sum_j(1-u_j)$ for $u_j\in[0,1]$ give
\[
 \D_B\geq\sum_s q_s\prod_{k=2}^n\langle e_s|\hat{\rho}_k|e_s\rangle
 \geq 1-\sum_{k=2}^n  (1-\mathcal V_{1k}).
\]
Thus, the following alternative condition
\begin{equation} 
\label{textV2}
\overline{\mathcal{V}}_{1*}:=\frac{1}{n-1}\sum_{k=2}^n \mathcal V_{1k}\ge 1- \frac{1}{2(n-1)}
\end{equation}
excludes an indistinguishability boost by Theorem~\ref{thm:noincrease}.

\paragraph{Conclusion.}
Adding invariant labels to particles can increase multiparticle indistinguishability in interference experiments, while simultaneously reducing pairwise indistinguishability.
Our seven-particle construction provides a quantitative example.
The enhancement reflects competition between reduced multiparticle amplitudes and canceled collective phases.
We also establish a general sufficient condition that excludes this effect.
Independent additional labels cannot increase indistinguishability when the original label state has a symmetric weight of at least $1/2$.
We give further sufficient conditions in terms of pairwise overlaps and fidelities with a reference state.
Our results clarify how collective phases, together with pairwise overlaps, determine multiparticle indistinguishability and identify regimes in which additional labels cannot increase it.

\begin{acknowledgments}
This work was supported by the National Council for Scientific and Technological Development (Conselho Nacional de Desenvolvimento Cient\'ifico e Tecnol\'ogico) of Brazil under Grant No.~307507/2023-8.
ChatGPT and Wolfram Mathematica were used for extensive symbolic calculations. ChatGPT was also used to check the results and proofread the manuscript.
\end{acknowledgments}

\onecolumngrid
\appendix
\section{Explicit formulas for the matrix permanents}
\label{appA}

\subsection{Temporal Gram matrix}
For $0\leq\rho<1$, choose an orthonormal basis $\{\ket{u_1},\ket{u_2}\}$ for the span of the temporal modes such that
\begin{equation}
 \ket{h_1}=\ket{u_1},\qquad
 \ket{h_2}=\rho\ket{u_1}+\sqrt{1-\rho^2}\ket{u_2}.
 \label{AE_u01}
\end{equation}
The formulas below extend continuously to $\rho=1$, where the temporal modes coincide.
In this orthonormal basis, the seven temporal states have the form
\begin{equation}
 \ket{f_j}=\alpha_j(\rho)\ket{u_1}+\beta_j(\rho)\ket{u_2},\qquad
 (\alpha_1,\beta_1)=(1,0),\quad
 (\alpha_2,\beta_2)=(\rho,\sqrt{1-\rho^2}).
 \label{AE_fk}
\end{equation}
For $k=3,\ldots,7$, let $d_k=\sqrt{2+2\rho\cos\theta_k}$. Then
\begin{equation}
 \alpha_k(\rho)=\frac{1+\rho e^{-i\theta_k}}{d_k},\qquad
 \beta_k(\rho)=\frac{\sqrt{1-\rho^2}\,e^{-i\theta_k}}{d_k}.
 \label{AE_alphbet}
\end{equation}
In the original pair of modes $\ket{h_1},\ket{h_2}$, the coefficient vectors are
\begin{equation}
 c_1=\begin{pmatrix}1\\0\end{pmatrix},\qquad
 c_2=\begin{pmatrix}0\\1\end{pmatrix},\qquad
 c_k=\frac{1}{d_k}\begin{pmatrix}1\\e^{-i\theta_k}\end{pmatrix}
 \quad (k=3,\ldots,7).
\end{equation}
Their inner products are computed using the Gram matrix of these modes,
\begin{equation}
 C_\rho=\begin{pmatrix}1&\rho\\ \rho&1\end{pmatrix},\qquad
 B_{ij}(\rho)=c_i^\dagger C_\rho c_j .
\end{equation}
In particular,
\begin{equation}
 B_{kl}(\rho)=
 \frac{1+\rho e^{-i\theta_l}+\rho e^{i\theta_k}+e^{i(\theta_k-\theta_l)}}{d_kd_l}
 \quad (k,l\geq3).
 \label{AE_Bkl}
\end{equation}
The other entries follow from Hermiticity and
\begin{equation}
 B_{11}=B_{22}=1,\qquad B_{12}=\rho,\qquad
 B_{1k}=\frac{1+\rho e^{-i\theta_k}}{d_k},\qquad
 B_{2k}=\frac{\rho+e^{-i\theta_k}}{d_k}\quad(k\geq3).
\end{equation}

\subsection{\texorpdfstring{Explicit formula for $\per B(\rho)$}{Permanent of the temporal Gram matrix}}
For a fixed occupation $m$ of $\ket{u_1}$, assign the seven photons to the ordered tensor product
$\ket{u_1}^{\otimes m}\otimes\ket{u_2}^{\otimes(7-m)}$.
Permutations within either block give the same assignment, so the distinct assignments are indexed by cosets
\begin{equation}
 S_7/S_{m,7-m},\qquad S_{m,7-m}:=S_m\times S_{7-m}.
\end{equation}
With one representative $\mu$ from each coset, the sum of the amplitudes in this sector is
\begin{equation}
 q_m(\rho)=\sum_{\mu\in S_7/S_{m,7-m}}
 \prod_{r=1}^{m}\alpha_{\mu(r)}(\rho)
 \prod_{r=m+1}^{7}\beta_{\mu(r)}(\rho).
 \label{eq:qm-coset}
\end{equation}
The terms are constant on each coset, whose size is $m!(7-m)!$. Expanding the symmetric projector in the orthonormal occupation basis therefore gives
\begin{equation}
 \per B(\rho)=\sum_{m=0}^{7}m!(7-m)!\,|q_m(\rho)|^2 .
 \label{eq:perB-coset}
\end{equation}
Introduce the generating function
\begin{equation}
 Q_B(x,y;\rho)=\prod_{j=1}^7\bigl(\alpha_j(\rho)x+\beta_j(\rho)y\bigr).
\end{equation}
The coefficient $[x^m y^{7-m}]Q_B$ equals $q_m(\rho)$, so
\begin{equation}
 \per B(\rho)=\sum_{m=0}^7 m!(7-m)!
          \left|[x^m y^{7-m}]Q_B(x,y;\rho)\right|^2 .
 \label{eq:perB-coefficient}
\end{equation}
Set $s=\sqrt{1-\rho^2}$ and $\zeta=e^{-2\pi i/5}$.
The five numbers $e^{-i\theta_k}$, $k=3,\ldots,7$, are the five distinct fifth roots of unity. Thus
\begin{equation}
 \begin{aligned}
 Q_B(x,y;\rho)
 &=\frac{x(\rho x+sy)}{\sqrt{\Delta(\rho)}}
          \prod_{a=0}^{4}\bigl(x+\zeta^a(\rho x+sy)\bigr)\\
 &=\frac{x(\rho x+sy)}{\sqrt{\Delta(\rho)}}
          \bigl(x^5+(\rho x+sy)^5\bigr),
 \end{aligned}
 \label{eq:QB-collapse}
\end{equation}
where
\begin{equation}
 \Delta(\rho)=\prod_{k=3}^7 d_k^2
             =2(1+\rho)(\rho^2+2\rho-4)^2 .
 \label{eq:Delta-rho}
\end{equation}
The polynomial $\widetilde Q_B(x,y;\rho)=x(\rho x+sy)\bigl(x^5+(\rho x+sy)^5\bigr)$ has real coefficients. Hence
\begin{equation}
 \per B(\rho)=\frac{1}{\Delta(\rho)}
 \sum_{m=0}^7 m!(7-m)!
       \left([x^m y^{7-m}]\widetilde Q_B(x,y;\rho)\right)^2 .
 \label{eq:perB-tildeQ}
\end{equation}
Expanding the binomials gives
\begin{equation}
 \begin{aligned}
 &\sum_{m=0}^7 m!(7-m)!
      \left([x^m y^{7-m}]\widetilde Q_B(x,y;\rho)\right)^2\\
 &\qquad=1440(1+\rho)
       \left(\rho^6-\rho^5+7\rho^4-7\rho^3+7\rho^2-\rho+1\right).
 \end{aligned}
 \label{eq:perB-numerator}
\end{equation}
Substitution into Eq.~\eqref{eq:perB-tildeQ} yields
\begin{equation}
 \per B(\rho)=
 \frac{720\left(\rho^6-\rho^5+7\rho^4-7\rho^3+7\rho^2-\rho+1\right)}
      {(\rho^2+2\rho-4)^2}.
 \label{eq:perB-rational}
\end{equation}
At $\rho=0$, this gives $\per B(0)=45$~\cite{Drury2016}.
At $\rho=1$, all temporal states coincide up to an overall phase for each photon, giving $\per B(1)=7!$.

\subsection{\texorpdfstring{Explicit formula for $\per(A\circ B(\rho))$}{Permanent of the combined Gram matrix}}
Write each polarization state as $\ket{\phi_j}=a_j\ket{H}+b_j\ket{V}$, where
\begin{equation}
 (a_1,b_1)=(1,0),\qquad (a_2,b_2)=(0,1),\qquad
 (a_k,b_k)=\frac{1}{\sqrt{2}}(1,e^{i\theta_k})\quad(k=3,\ldots,7).
\end{equation}
In the ordered basis $\ket{H}\otimes\ket{u_1},\ket{H}\otimes\ket{u_2},
\ket{V}\otimes\ket{u_1},\ket{V}\otimes\ket{u_2}$, the state
$\ket{\phi_j}\otimes\ket{f_j}$ has amplitudes
\begin{equation}
 (\gamma_{j,11},\gamma_{j,12},\gamma_{j,21},\gamma_{j,22})
 =(a_j\alpha_j,\,a_j\beta_j,\,b_j\alpha_j,\,b_j\beta_j).
\end{equation}
For an occupation vector
\begin{equation}
 \bm m=(m_{11},m_{12},m_{21},m_{22}),\qquad
 m_{ab}\in\mathbb Z_{\geq0},\qquad |\bm m|:=\sum_{a,b=1}^2m_{ab}=7,
\end{equation}
the distinct assignments are indexed by the cosets
\begin{equation}
 S_7/S_{\bm m},\qquad S_{\bm m}:=S_{m_{11}}\times S_{m_{12}}\times S_{m_{21}}\times S_{m_{22}} .
\end{equation}
The corresponding amplitude is
\begin{equation}
 \begin{aligned}
 q_{\bm m}(\rho)
 =\sum_{\mu\in S_7/S_{\bm m}}&
 \left(\prod_{r=1}^{m_{11}}\gamma_{\mu(r),11}\right)
 \left(\prod_{r=m_{11}+1}^{m_{11}+m_{12}}\gamma_{\mu(r),12}\right)\\
 {}\times&
 \left(\prod_{r=m_{11}+m_{12}+1}^{m_{11}+m_{12}+m_{21}}\gamma_{\mu(r),21}\right)
 \left(\prod_{r=m_{11}+m_{12}+m_{21}+1}^{7}\gamma_{\mu(r),22}\right) .
 \end{aligned}
\end{equation}
As in Eq.~\eqref{eq:perB-coset}, permutations within each occupation block contribute a factorial factor, giving
\begin{equation}
 \per(A\circ B(\rho))
 =\sum_{|\bm m|=7}
    \left(\prod_{a,b=1}^{2}m_{ab}!\right)|q_{\bm m}(\rho)|^2 .
 \label{eq:perG-coset}
\end{equation}
For the generating function
\begin{equation}
 Q_{A\circ B}(\bm x;\rho)
 =\prod_{j=1}^7\bigl(a_j\alpha_j x_{11}+a_j\beta_j x_{12}
                     +b_j\alpha_j x_{21}+b_j\beta_j x_{22}\bigr),
\end{equation}
$q_{\bm m}(\rho)$ is the coefficient of $\bm x^{\bm m}:=\prod_{a,b=1}^2x_{ab}^{m_{ab}}$. Hence
\begin{equation}
 \per(A\circ B(\rho))=\sum_{|\bm m|=7}
   \left(\prod_{a,b=1}^2m_{ab}!\right)
   \left|[\bm x^{\bm m}]Q_{A\circ B}(\bm x;\rho)\right|^2 .
 \label{eq:perG-coefficient}
\end{equation}
Evaluating the coefficient sum gives
\begin{equation}
 \per(A\circ B(\rho))=
 \frac{5\left(1237-1237\rho-968\rho^2+968\rho^3+232\rho^4-160\rho^5\right)}
      {8(\rho^2+2\rho-4)^2} .
 \label{eq:perAB-rational}
\end{equation}
Appendix~\ref{appB} gives an independent evaluation of both permanents by cycle type.

\subsection{Condition for an indistinguishability boost}\label{app:threshold}
Subtracting Eq.~\eqref{eq:perB-rational} from Eq.~\eqref{eq:perAB-rational} gives
\begin{equation}
 \begin{aligned}
 \per(A\circ B(\rho))-\per B(\rho)
 &=-\frac{5p(\rho)}{8(\rho^2+2\rho-4)^2},\\
 p(\rho)&=1152\rho^6-992\rho^5+7832\rho^4-9032\rho^3+9032\rho^2+85\rho-85 .
 \end{aligned}
 \label{eq:per-difference-rational}
\end{equation}
The denominator is positive on $[0,1]$. To locate the crossing, write
\begin{equation*}
 \frac{p(\rho)}{1-\rho}
 =-85+9032\rho^2+992\rho^4
      +\frac{6840\rho^4+1152\rho^6}{1-\rho},
 \qquad 0\leq\rho<1 .
\end{equation*}
The right-hand side increases strictly from $-85$ to $+\infty$.
Hence $p$ has exactly one zero $\rho_*=0.096583467644\ldots$ on this interval, and the permanent difference is positive precisely for $0\leq\rho<\rho_*$.
At the remaining endpoint, $\per B(1)=7!$ and $\per(A\circ B(1))=\per A=45$, as required when the temporal states coincide up to phase.

\section{Cycle expansions of the matrix permanents}
\label{appB}

\subsection{Cycle-type generating function}
A partition $\lambda=1^{m_1}2^{m_2}\cdots n^{m_n}$ of $n$ records the number $m_k$ of cycles of each length $k$, with $\sum_k k m_k=n$.
With $p_k$ marking a cycle of length $k$, the cycle index of the symmetric group is
\begin{equation}
 Z_n(\bm p)=\frac{1}{n!}\sum_{\sigma}\prod_{k=1}^n p_k^{m_k(\sigma)} .
 \label{Zn}
\end{equation}
There are $\mathcal N_\lambda=n!/\prod_{k=1}^n k^{m_k}m_k!$ permutations of type $\lambda$. Thus the generating function of the cycle indices, with $Z_0=1$, is
\begin{equation}
 \sum_{n\geq0}Z_n(\bm p)t^n=\exp\left(\sum_{k\geq1}\frac{p_kt^k}{k}\right).
 \label{Z_n}
\end{equation}

For an $n\times n$ matrix $M$, introduce auxiliary variables $x_1,\ldots,x_n$ and the diagonal matrix $X=\operatorname{diag}(x_1,\ldots,x_n)$.
The analogous expansion of the permanent is
\begin{equation}
 \per M=[x_1\cdots x_n]\exp\left(
   \sum_{k=1}^n\frac{1}{k}\sum_{i_1,\ldots,i_k=1}^n
   M_{i_1i_2}\cdots M_{i_ki_1}x_{i_1}\cdots x_{i_k}\right).
\end{equation}
The inner sum is $\Tr[(XM)^k]$. Extracting $x_1\cdots x_n$ selects products of disjoint cycles that use every index exactly once.
The factor $1/k$ cancels the $k$ cyclic rotations of each cycle, and the factorials in the exponential account for reorderings of cycles.
Marking each cycle length as in Eq.~\eqref{Z_n} gives
\begin{equation}
 \Phi_M(\bm p,\bm x)=\exp\left(\sum_{k=1}^n\frac{p_k}{k}\Tr[(XM)^k]\right).
\end{equation}
The class sum $d^\lambda(M)$ of Eq.~\eqref{perG} is therefore
\begin{equation}
 \begin{aligned}
 d^\lambda(M)
 &=[x_1\cdots x_n][p_1^{m_1}\cdots p_n^{m_n}]\Phi_M(\bm p,\bm x)\\
 &=[x_1\cdots x_n]\prod_{k=1}^n\frac{1}{m_k!}
              \left(\frac{\Tr[(XM)^k]}{k}\right)^{m_k}.
 \end{aligned}
 \label{eq:class-coeff}
\end{equation}
Summing over the partitions gives $\per M=\sum_{\lambda\vdash n}d^\lambda(M)$.

\subsection{Comparison of the temporal and combined permanents}
Set $d_1=d_2=1$ and use $d_k=\sqrt{2+2\rho\cos\theta_k}$ for $k=3,\ldots,7$ as in Appendix~\ref{appA}.
The entries of $B(\rho)$ are $B_{ij}=N_{ij}(\rho)/(d_id_j)$, where each numerator is at most linear in $\rho$.
Every permutation term uses every row and column once, so both permanents have the common denominator
\[
 \prod_{i=1}^7d_i^2=\Delta(\rho)=2(1+\rho)(\rho^2+2\rho-4)^2 .
\]
Equations~\eqref{eq:perB-rational} and~\eqref{eq:perAB-rational} can be written as
\begin{equation}
 \per B(\rho)=\frac{P_B(\rho)}{(\rho^2+2\rho-4)^2},\qquad
 \per(A\circ B(\rho))=\frac{P_{A\circ B}(\rho)}{(\rho^2+2\rho-4)^2},
 \label{AE_perABN}
\end{equation}
with polynomial numerators $P_B$ and $P_{A\circ B}$.

\subsubsection{Polynomials for the individual cycle types}
Applying Eq.~\eqref{eq:class-coeff} gives the class polynomials listed below.
In this example, the factor $2(1+\rho)$ cancels in each class sum separately. For $M=B$ or $A\circ B$, write
\begin{equation}
 P_M(\rho)=\sum_{\lambda\vdash7}s_M^\lambda(0)P_M^\lambda(\rho),
 \qquad s_M^\lambda(0)=\operatorname{sgn}d^\lambda(M(0)),
 \qquad P_M^\lambda(0)>0 .
 \label{PlambdaM}
\end{equation}
The signed class contribution to the permanent is
\begin{equation}
 d^\lambda(M(\rho))=
 \frac{s_M^\lambda(0)P_M^\lambda(\rho)}{(\rho^2+2\rho-4)^2}.
 \label{dlambdaP}
\end{equation}
The middle column of each table gives the sign at $\rho=0$. At other overlaps, the sign follows from the full polynomial in Eq.~\eqref{dlambdaP}.
Classes consisting only of fixed points and transpositions have nonnegative contributions for every $\rho$.

\[
\boxed{
\begin{array}{c|c|c}
\lambda & s_B^\lambda(0) & P^\lambda_B(\rho)\\
\hline
(7) & - &
60(1+\rho)
\left(3-6\rho+8\rho^2-10\rho^3+\rho^4-2\rho^5\right)
\\[2pt]
(6,1) & - &
20
\left(17-17\rho-43\rho^2+43\rho^3-64\rho^4+28\rho^5-6\rho^6\right)
\\[2pt]
(5,2) & - &
12(1+\rho)
\left(5-10\rho-17\rho^2+44\rho^3-36\rho^4-7\rho^5\right)
\\[2pt]
(5,1,1) & - &
12
\left(18-18\rho-127\rho^2+127\rho^3-57\rho^4+20\rho^5-5\rho^6\right)
\\[2pt]
\hline
(4,3) & + &
70(1+\rho)
\left(1-2\rho-4\rho^2+10\rho^3-3\rho^4+\rho^5\right)
\\[2pt]
(4,2,1) & + &
90
\left(1-\rho+7\rho^2-7\rho^3+7\rho^4-\rho^5+\rho^6\right)
\\[2pt]
(4,1,1,1) & + &
10
\left(2-2\rho+105\rho^2-105\rho^3+14\rho^4+5\rho^5+2\rho^6\right)
\\[2pt]
(3,3,1) & + &
20
\left(6-6\rho-9\rho^2+9\rho^3+33\rho^4-21\rho^5+2\rho^6\right)
\\[2pt]
(3,2,2) & + &
5(1+\rho)
\left(18-36\rho+62\rho^2-88\rho^3+58\rho^4+7\rho^5\right)
\\[2pt]
(3,2,1,1) & + &
10
\left(40-40\rho+56\rho^2-56\rho^3+70\rho^4-33\rho^5+5\rho^6\right)
\\[2pt]
(3,1,1,1,1) & + &
5
\left(28-28\rho+74\rho^2-74\rho^3+4\rho^4+9\rho^5+\rho^6\right)
\\[2pt]
(2,2,2,1) & + &
5
\left(26-26\rho+20\rho^2-20\rho^3-\rho^4+19\rho^5+3\rho^6\right)
\\[2pt]
(2,2,1,1,1) & + &
5
\left(60-60\rho+34\rho^2-34\rho^3+27\rho^4-8\rho^5+2\rho^6\right)
\\[2pt]
(2,1,1,1,1,1) & + &

 140-140\rho+46\rho^2-46\rho^3+11\rho^4+9\rho^5+\rho^6
\\[2pt]
(1,1,1,1,1,1,1) & + &
 (\rho^2+2\rho-4)^2
\end{array}
}
\]

\[
\boxed{ \begin{array}{c|c|c}
\lambda & s_{A\circ B}^\lambda(0) & P^\lambda_{A\circ B}(\rho)\\
 \hline
(7) &+ &
\dfrac{5}{16}(1+\rho)
\left(81-162\rho+32\rho^2+98\rho^3-67\rho^4\right)
\\[2pt]
(6,1) & + &
\dfrac{5}{16}
\left(219-219\rho-362\rho^2+362\rho^3-21\rho^4-47\rho^5\right)
\\[2pt]
(5,2) & + &
\dfrac{5}{16}(1+\rho)
\left(72-144\rho+29\rho^2+86\rho^3-49\rho^4\right)
\\[2pt]
(5,1,1) & + &
\dfrac{1}{16}
\left(1449-1449\rho-2371\rho^2+2371\rho^3+14\rho^4-230\rho^5\right)
\\[2pt]
(4,3) & + &
\dfrac{5}{16}(1+\rho)
\left(51-102\rho+27\rho^2+48\rho^3-17\rho^4\right)
\\[2pt]
(4,2,1) & + &
\dfrac{5}{16}
\left(199-199\rho-229\rho^2+229\rho^3+64\rho^4-46\rho^5\right)
\\[2pt]
(4,1,1,1) & + &
\dfrac{5}{8}
\left(128-128\rho-140\rho^2+140\rho^3+23\rho^4-21\rho^5\right)
\\[2pt]
(3,3,1) & + &
\dfrac{5}{8}
\left(39-39\rho-37\rho^2+37\rho^3+24\rho^4-12\rho^5\right)
\\[2pt]
(3,2,2) & + &
\dfrac{5}{16}(1+\rho)
\left(42-84\rho+64\rho^2-44\rho^3+31\rho^4\right)
\\[2pt]
(3,2,1,1) & + &
\dfrac{5}{16}
\left(333-333\rho-141\rho^2+141\rho^3+128\rho^4-48\rho^5\right)
\\[2pt]
(3,1,1,1,1) & + &
\dfrac{5}{8}
\left(108-108\rho-19\rho^2+19\rho^3+18\rho^4-4\rho^5\right)
\\[2pt]
(2,2,2,1) & + &
\dfrac{5}{16}
\left(66-66\rho+\rho^2-\rho^3-9\rho^4+35\rho^5\right)
\\[2pt]
(2,2,1,1,1) & + &
\dfrac{25}{8}
\left(28-28\rho+\rho^2-\rho^3+6\rho^4\right)
\\[2pt]
(2,1,1,1,1,1) & + &
\dfrac{5}{4}
\left(60-60\rho+\rho^2-\rho^3+5\rho^4+2\rho^5\right)
\\[2pt]
(1,1,1,1,1,1,1) & + &
 (\rho^2+2\rho-4)^2
\end{array}
}
\]

\section{Fidelity and pairwise visibility tests}
\label{appC}

Multiplying the two density operators in Eq.~\eqref{1B-labels} and taking the trace of their product gives the exact identity
\begin{equation}
 \mathcal V_{kl}=\mathcal F_k\mathcal F_l
      +(1-\mathcal F_k)(1-\mathcal F_l)\Tr(\hat{\rho}_k^\perp\hat{\rho}_l^\perp).
\end{equation}
This proves Eq.~\eqref{Vkl}. Suppose the reference state has weight at least $1/2$ in every label state, and set $x_k=2\mathcal F_k-1\geq0$.
The upper bound in Eq.~\eqref{Vkl} gives $x_kx_l\geq(2\mathcal V_{kl}-1)_+$, where $(z)_+=\max(z,0)$.
By the arithmetic--geometric mean inequality,
\begin{equation}
 \mathcal F_k\mathcal F_l
 =\frac{1+x_k+x_l+x_kx_l}{4}
 \geq\left(\frac{1+\sqrt{x_kx_l}}{2}\right)^2
 \geq \left[h(\mathcal V_{kl})\right]^2,\qquad	
 h(v):=\frac{1+\sqrt{(2v-1)_+}}{2}.
\label{ineqFkFl}\end{equation}
Each fidelity $\mathcal{F}_i$ appears in exactly $n-1$ products on the left-hand side of Eq.~\eqref{ineqFkFl}. Multiplying these inequalities over all pairs and taking the $(n-1)$th root gives the sufficient test
\begin{equation}
 \prod_{k=1}^{n}\mathcal F_k
 \ \geq\ \prod_{k<l}\left[h(\mathcal V_{kl})\right]^{\frac{2}{n-1}} .
 \label{eq:heterogeneous-visibility}
\end{equation}
Thus a right-hand side of at least $1/2$ guarantees Eq.~\eqref{DboundF}.

A test based on the minimum pairwise visibility follows by replacing the right-hand side of Eq.~\eqref{eq:heterogeneous-visibility} with $\left[h(\min_{k<l}\mathcal V_{kl})\right]^n$, yielding Eq.~\eqref{eq:uniform-visibility}. Under this test, knowing that $\mathcal F_i\geq1/2$ for one particle fixes the sign of every $x_k$. In fact, Eq.~\eqref{eq:uniform-visibility} gives $\mathcal V_{ik}>1/2$, so $x_ix_k>0$ for all $k\ne i$.

A simpler condition applies if one single-particle label state is pure and is chosen as the reference, so that $\hat{\Pi}=\hat{\rho}_i$. This choice gives $\mathcal F_i=1$ and $\mathcal F_k=\mathcal V_{ik}$ for $k\ne i$, even when the other density operators have coherences in this basis. Consequently, $\prod_{k\ne i}\mathcal V_{ik}\geq1/2$ suffices by Eq.~\eqref{eq:fidelity-weight}.


\begin{thebibliography}{99}
\bibitem{LOC} E. Knill, R. Laflamme, and G. J. Milburn, A scheme for efficient quantum computation with linear optics, Nature \textbf{409}, 46 (2001).
\bibitem{AA} S. Aaronson and A. Arkhipov, The computational complexity of linear optics, Theory Comput. \textbf{9}, 143 (2013).
\bibitem{RevLOC} P. Kok, W. J. Munro, K. Nemoto, T. C. Ralph, J. P. Dowling, and G. J. Milburn, Linear optical quantum computing with photonic qubits, Rev. Mod. Phys. \textbf{79}, 135 (2007).
\bibitem{HOM} C. K. Hong, Z. Y. Ou, and L. Mandel, Measurement of subpicosecond time intervals between two photons by interference, Phys. Rev. Lett. \textbf{59}, 2044 (1987).
\bibitem{QDsHOM} C. Santori, D. Fattal, J. Vu\v{c}kovi\'c, G. S. Solomon, and Y. Yamamoto, Indistinguishable photons from a single-photon device, Nature \textbf{419}, 594 (2002).
\bibitem{QDsHOM2} N. Margaria, F. Pastier, T. Bennour, M. Billard, E. Ivanov, W. Hease \textit{et al.}, Efficient fibre-pigtailed source of indistinguishable single photons, Nat. Commun. \textbf{16}, 7553 (2025).
\bibitem{R2Ds} D. G\'erard, S. Buil, K. Watanabe, T. Taniguchi, J.-P. Hermier, and A. Delteil, Resonance fluorescence and indistinguishable photons from a coherently driven B centre in hBN, Nat. Commun. \textbf{17}, 1843 (2026).
\bibitem{QDs} P. Michler, A. Kiraz, C. Becher, W. V. Schoenfeld, P. M. Petroff, L. Zhang, E. Hu, and A. Imamoglu, A Quantum Dot Single-Photon Turnstile Device, Science \textbf{290}, 2282 (2000).

\bibitem{3phPhase} A. J. Menssen, A. E. Jones, B. J. Metcalf, M. C. Tichy, S. Barz, W. S. Kolthammer, and I. A. Walmsley, Distinguishability and many-particle interference, Phys. Rev. Lett. \textbf{118}, 153603 (2017).
\bibitem{nphPhases} V. S. Shchesnovich and M. E. O. Bezerra, Collective phases of identical particles interfering on linear multiports, Phys. Rev. A \textbf{98}, 033805 (2018).
\bibitem{4phphase} A. E. Jones, A. J. Menssen, H. M. Chrzanowski, T. A. W. Wolterink, V. S. Shchesnovich, and I. A. Walmsley, Multiparticle Interference of Pairwise Distinguishable Photons, Phys. Rev. Lett. \textbf{125}, 123603 (2020).
\bibitem{Drury2016} S. W. Drury, A counterexample to a question of Bapat and Sunder, Electron. J. Linear Algebra \textbf{31}, 69 (2016).
 \bibitem{BCount} B. Seron, L. Novo, and N. J. Cerf, Boson bunching is not maximized by indistinguishable particles, Nat. Photonics \textbf{17}, 702 (2023).
\bibitem{UnfBunch2026} L. Pioge, L. Novo, and N. J. Cerf, A unified framework for anomalous boson bunching, \href{https://arxiv.org/abs/2607.19499}{arXiv:2607.19499} (2026).
\bibitem{Shch2015} V. S. Shchesnovich, Partial indistinguishability theory for multiphoton experiments in multiport devices, Phys. Rev. A \textbf{91}, 013844 (2015).
 \bibitem{Shch2015A} V. S. Shchesnovich, Tight bound on the trace distance between a realistic device with partially indistinguishable bosons and the ideal BosonSampling, Phys. Rev. A \textbf{91}, 063842 (2015).
 \bibitem{Tichy2015} M. C. Tichy, Sampling of partially distinguishable bosons and the relation to the multidimensional permanent, Phys. Rev. A \textbf{91}, 022316 (2015). 
 \bibitem{IntJQuntInf2015} V. S. Shchesnovich, Boson-sampling with non-interacting fermions, Int. J. Quantum Inf. \textbf{13}, 1550013 (2015).
\end{thebibliography}
\end{document}